\documentclass[a4paper,11pt]{article}

\newcommand{\longversion}[1]{#1}
\newcommand{\shortversion}[1]{}
\newcommand{\longshort}[2]{\longversion{#1}\shortversion{#2}}
 
\longversion{\usepackage{amsthm}}

\usepackage{amsmath,amssymb,amstext}
\usepackage{thm-restate} 

\usepackage{tikz}
\usepackage{boxedminipage}
\usepackage{xspace}

\longversion{
 \usepackage{vmargin}
 \setpapersize{USletter}
 \setmarginsrb{1in}{1in}{1in}{0.5in}{0cm}{0cm}{0.25in}{0.25in}
}

\longshort{
 \newtheorem{definition}{Definition}
 \newtheorem{lemma}{Lemma} 
 
 \newtheorem{theorem}{Theorem}

 \newtheorem{fact}{Fact}
 \newtheorem{myrule}{Rule}
 \newcommand{\myqed}{}
}{
 \spnewtheorem{observation}{Observation}{\bfseries}{\itshape}
 \spnewtheorem{myrule}{Rule}{\bfseries}{\itshape}
 \spnewtheorem{fact}{Fact}{\bfseries}{\itshape}
 \newcommand{\myqed}{\qed}
}

\newcommand{\Card}[1]{|#1|}

\newcommand{\cC}{\mathcal{C}}

\newcommand{\cO}{\mathcal{O}}

\newcommand{\NP}{\text{\normalfont NP}}

\newcommand{\FPT}{\text{\normalfont FPT}}

\newcommand{\fpt}{fixed-pa\-ra\-me\-ter trac\-ta\-ble\xspace}

\newcommand{\set}[1]{\left\{ #1 \right\}}
\newcommand{\myiff}{if and only if\xspace}

\newcommand{\ie}{i.e.\xspace}
\newcommand{\etal}{\emph{et al.}\xspace}

\newcommand{\sbn}{\mathbf{sb}_{\mathbf{N}}}
\newcommand{\tw}{{\mathbf{tw}}}

\newcommand{\var}{{\normalfont \textsf{var}}}
\newcommand{\cla}{{\normalfont \textsf{cla}}}
\newcommand{\lit}{{\normalfont \textsf{lit}}}
\newcommand{\true}{1} 
\newcommand{\false}{0} 

\newcommand{\inc}{{\normalfont \textsf{inc}}}
\newcommand{\incu}{{\normalfont \textsf{inc+u}}}
\newcommand{\univ}{{\normalfont \textsf{univ}}}

\newcommand{\Nested}{\textsc{Nested}\xspace}
\newcommand{\RHorn}{\textsc{RHorn}\xspace}
\newcommand{\BDS}{back\-door set\xspace}
\newcommand{\BDSs}{back\-door sets\xspace}

\newcommand{\NBDS}{\Nested-\BDS}
\newcommand{\NBDSs}{\Nested-\BDSs}
\newcommand{\SNBDS}{strong \NBDS}
\newcommand{\SNBDSs}{strong \NBDSs}

\newcommand{\obs}{\Nested-ob\-struc\-tion\xspace}
\newcommand{\obss}{\Nested-ob\-struc\-tions\xspace}

\newcommand{\ktw}{\mathsf{tw}(k)}
\newcommand{\kgrid}{\mathsf{grid}(k)}
\newcommand{\kobs}{\mathsf{obs}(k)}
\newcommand{\ksame}{\mathsf{same}(k)}

\title{Strong Backdoors to Nested Satisfiability \thanks{%
The authors acknowledge support from the European Research Council
(COMPLEX REASON, 239962).%
}}

\author{%
Serge Gaspers \and 
Stefan Szeider
}

\longshort{
 \date{%
 Institute of Information Systems\\ Vienna University of Technology\\ Vienna, Austria.\\
 \texttt{gaspers@kr.tuwien.ac.at}\\ \texttt{stefan@szeider.net}
 }
}{
 \institute{Institute of Information Systems, Vienna University of Technology, Vienna, Austria.\\
 \texttt{gaspers@kr.tuwien.ac.at}, \texttt{stefan@szeider.net}}
}

\begin{document}

\maketitle

\begin{abstract}
  Knuth (1990) introduced the class of nested formulas and showed that
  their satisfiability can be decided in polynomial time.  We show that,
  parameterized by the size of a smallest strong backdoor set to the
  base class of nested formulas, checking the satisfiability of any
  CNF formula is fixed-parameter tractable.  Thus, for any $k>0$, the
  satisfiability problem can be solved in polynomial time for any
  formula~$F$ for which there exists a variable set $B$ of size at most
  $k$ such that for every truth assignment $\tau$ to $B$, the formula
  $F[\tau]$ is nested; moreover, the degree of the polynomial is
  independent of $k$. 

 Our algorithm uses the grid-minor theorem of Robertson and Seymour
 (1986) to either find that the incidence graph of the formula has
 bounded treewidth---a case that is solved using  model
 checking for monadic second order logic---or to find many
 vertex-disjoint obstructions in the incidence graph. For the latter
 case, new combinatorial arguments are used to find a small backdoor
 set. Combining both cases leads to an approximation algorithm
 producing a strong backdoor set whose size is upper bounded by a
 function of the optimum. Going through all assignments to this set of
 variables and using Knuth's algorithm, the satisfiability of the input
 formula is decided.
\end{abstract}

\section{Introduction}

In a 1990 paper \cite{Knuth90} Knuth introduced the class of nested
CNF formulas and showed that their satisfiability can be decided in
polynomial time.  A CNF formula is \emph{nested} if its variables can
be linearly ordered such that there is no pair of clauses that
\emph{straddle} each other; a clause $c$ straddles a clause $c'$ if
there are variables $x,y \in \var(c)$ and $z\in \var(c')$ such that
$x<z<y$ in the linear ordering under consideration.  \Nested denotes
the class of nested CNF formulas.  For an example see
Figure~\ref{fig:nested}.
\begin{figure}[tbh]
  \centering

\tikzset{var/.style={inner sep=.15em,circle,fill=black,draw},
         clause/.style={minimum size=1mm,rectangle,fill=white,draw},
         label distance=-1pt}
 \centering

\longversion{
  \begin{tikzpicture}
   \node (1) at (1,0) [var,label=below:$\strut t$] {};
   \node (2) at (2,0) [var,label=below:$\strut u$] {};
   \node (3) at (3,0) [var,label=below:$\strut v$] {};
   \node (4) at (4,0) [var,label=below:$\strut w$] {};
   \node (5) at (5,0) [var,label=below:$\strut x$] {};
   \node (6) at (6,0) [var,label=below:$\strut y$] {};
   \node (7) at (7,0) [var,label=below:$\strut z$] {};

   \node (s) at (1.5,0.4) [clause,label=above:$c_1$] {};
   \node (t) at (3,0.4) [clause,label=above:$c_2$] {};
   \node (u) at (4.5,0.4) [clause,label=above:$c_3$] {};
   \node (v) at (5.5,0.4) [clause,label=above:$c_4$] {};
   \node (w) at (6.5,0.4) [clause,label=above:$c_5$] {};

   \node (x) at (2,1.3) [clause,label=above:$c_6$] {};
   \node (y) at (6,1.3) [clause,label=above:$c_7$] {};

   \node (z) at (3,2.5) [clause,label=above:$c_8$] {};

   \draw (1)--(s)--(2)--(t)--(4)--(u)--(5)--(v)--(6)--(w)--(7)
   (3)--(t)
   (1) ..controls +(.05,.8) ..(x)--(2) 
   (4) .. controls +(-.3,.8).. (x)
   (5) ..controls +(.05,.8) .. (y)
   (7) ..controls +(-.05,.8) .. (y)
   
   (1) .. controls +(0,1.5)  .. (z) 

   (5) .. controls +(0,1.5) .. (z) 
   (4) .. controls +(0,.8) .. (z)
;
\end{tikzpicture}

  }
\shortversion{
\begin{tikzpicture}
   \node (1) at (1,0) [var,label=below:$\strut t$] {};
   \node (2) at (2,0) [var,label=below:$\strut u$] {};
   \node (3) at (3,0) [var,label=below:$\strut v$] {};
   \node (4) at (4,0) [var,label=below:$\strut w$] {};
   \node (5) at (5,0) [var,label=below:$\strut x$] {};
   \node (6) at (6,0) [var,label=below:$\strut y$] {};
   \node (7) at (7,0) [var,label=below:$\strut z$] {};

   \node (s) at (1.5,0.4) [clause,label=above:$c_1$] {};
   \node (t) at (3,0.4) [clause,label=above:$c_2$] {};
   \node (u) at (4.5,0.4) [clause,label=above:$c_3$] {};
   \node (v) at (5.5,0.4) [clause,label=above:$c_4$] {};
   \node (w) at (6.5,0.4) [clause,label=above:$c_5$] {};

   \node (x) at (2,1.5) [clause,label=above:$c_6$] {};
   \node (y) at (6,1.5) [clause,label=above:$c_7$] {};

   \node (z) at (3,2.5) [clause,label=above:$c_8$] {};

   \draw (1)--(s)--(2)--(t)--(4)--(u)--(5)--(v)--(6)--(w)--(7)
   (3)--(t)
   (1) ..controls +(.1,.5) ..(x)--(2) 
   (4) .. controls +(-.3,.5).. (x)
   (5) ..controls +(.1,.5) .. (y)
   (7) ..controls +(-.1,.5) .. (y)
   
   (1) .. controls +(0,1.5)  .. (z) 

   (5) .. controls +(0,1.5) .. (z) (4)--(z)
;
  
  \end{tikzpicture}

}

  \caption{Incidence graph of the nested formula $F=\bigwedge_{i=1}^8 c_i$
    with 
    $c_1=t \vee \neg u$,
    $c_2=u \vee v  \vee w$,
    $c_3=w \vee x$,
    $c_4=x \vee \neg y$,
    $c_5= y  \vee  \neg z$,
    $c_6= t \vee u \vee \neg w$,
    $c_7= \neg x \vee z$,
    $c_8= \neg t \vee w \vee x$.}
  \label{fig:nested}
\end{figure}
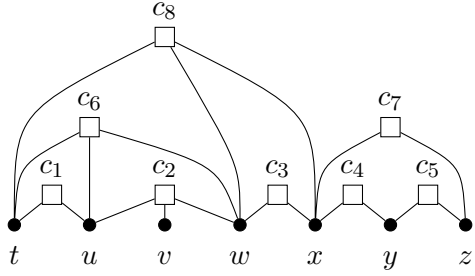
Since nested formulas have incidence graphs of bounded
treewidth~\cite{BiedlH04}, one can use treewidth-based
algorithms~\cite{FischerMakowskyRavve06,SamerSzeider10} to even
compute the number of satisfying truth assignments of nested formulas
in polynomial time (incidence graphs are defined in
Section~\ref{section:prelims}).  Hence the problems SAT and \#SAT are
polynomial for nested formulas.

The aim of this paper is to extend the nice computational properties of
nested formulas to formulas that are not nested but are of small
distance from being nested. We measure the distance of a CNF formula $F$
from being nested as the size of a smallest set $B$ of variables, such
that for all partial truth assignments $\tau$ to $B$, the reduced
formula $F[\tau]$ is nested. Such a set $B$ is called a \emph{strong
  backdoor set} with respect to the class of nested
formulas~\cite{WilliamsGomesSelman03}, or strong \Nested-backdoor set,
for short. Once we have found such a backdoor set of size $k$, we can
decide the satisfiability of $F$ by checking the satisfiability of $2^k$
nested formulas, or for model counting, we can take the sum of the
number of models of the $2^k$ nested formulas. Thus the problems SAT and
\#SAT can be solved in time $O(2^k\Card{F}^c)$ where $\Card{F}$ denotes
the length of $F$ and $k$ denotes the size of the given strong
\Nested-backdoor set; $c$ is a small constant.  In other words, the
problems SAT and \#SAT are \emph{fixed-parameter tractable} for
parameter~$k$ (for background on fixed-parameter tractability see
Section~\ref{section:prelims}). However, in order to use the backdoor
set we must find it first. Is the detection of strong \Nested-backdoor
sets fixed-parameter tractable as well?

Let $\sbn(F)$ denote the size of a smallest strong \Nested-backdoor
set of a CNF formula~$F$. To find a strong backdoor set of size
$k=\sbn(F)$ one can try all possible sets of variables of size at most
$k$, and check for each set whether it is a strong backdoor
set. However, for a formula with $n$ variables we have to check
$\binom{n}{k}=\Omega(n^k)$ such sets. Thus, this brute-force approach
scales poorly in~$k$ and does not provide fixed-parameter
tractability, as the order of the polynomial increases with~$k$.

In this paper we show that one can overcome this limitation with a
more sophisticated algorithm.  \emph{We show that the problems SAT and
  \#SAT are fixed-parameter tractable when parameterized by $\sbn$,
  the size of a smallest strong \Nested-backdoor set, even when the
  backdoor set is not provided as an input.}

Our algorithm is constructive and uses the Grid Minor Theorem of
Robertson and Seymour~\cite{RobertsonSeymour86b} to either find that
the incidence graph of the formula has bounded treewidth---a case that
is solved using model checking for monadic second order
logic~\cite{ArnborgLagergrenSeese91}---or to find many vertex-disjoint
obstructions in the incidence graph. For the latter case, new
combinatorial arguments are used to find a small strong backdoor set.
Combining both cases leads to an algorithm producing a strong backdoor
set of a given formula $F$ of size at most $2^k$ for
$k=\sbn(F)$. Solving all the $2^{2^{k}}$ resulting nested formulas
provides a solution to~$F$.

Our result provides a new parameter $\sbn$ that makes SAT and \#SAT
fixed-parameter tractable. The parameter $\sbn$ is \emph{incomparable}
with other known parameters that make SAT and \#SAT fixed-parameter
tractable.  Take for instance the treewidth of the incidence graph of
a CNF formula $F$, denoted $\tw^*(F)$.  As mentioned above, SAT and
\#SAT are fixed-parameter tractable for parameter
$\tw^*$~\cite{FischerMakowskyRavve06,SamerSzeider10}, and
$\tw^*(F)\leq 3$ holds if $\sbn(F)=0$ (i.e., if $F\in
\Nested$)~\cite{BiedlH04}. However, by allowing only $\sbn(F)=1$ we
already get formulas with arbitrarily large $\tw^*(F)$. This can be
seen as follows. Take a CNF formula $F_n$ whose incidence graph is an
$n\times n$ square grid, with vertices $v_{i,j}$, $1\leq i,j\leq n$.
Assume the clauses correspond to vertices $v_{i,j}$ where $i+j$ is
even, and call a clause even or odd according to whether for the
corresponding vertex $v_{i,j}$, the sum $i+j$ is a multiple of $4$ or
not, respectively.  It is well known that the $n\times n$ grid, $n\ge 2$, has
treewidth $n$ (folklore). 
Hence we have $\tw^*(F)=n$. Now
take a new variable $x$ and add it positively to all odd clauses and
negatively to all even clauses. Let $F_n^x$ denote the new
formula. Since the incidence graph of $F_n$ is a subgraph of the
incidence graph of $F_n^x$, we have $\tw^*(F_n^x)\geq
\tw^*(F_n)=n$. However, setting~$x$ to true removes all odd clauses
and thus yields a formula whose incidence graph is a disjoint union of
paths, which is easily seen to be nested. Similarly, setting~$x$ to
false yields a nested formula as well. Hence $\{x\}$ forms a strong
\Nested-backdoor set, and so $\sbn(F)=1$. One can also construct
formulas where $\sbn$ is large and $\tw^*$ is small, for example by taking the
variable-disjoint union $F$ of formulas $F_i = (x_i \vee y_i \vee z_i) \wedge (\neg x_i \vee y_i \vee z_i)$ with $\sbn(F_i)=1$ and
$\tw^*(F_i)=2$, $1\leq i \leq n$.  Then $\tw^*(F)=\tw^*(F_i)=2$, but
$\sbn(F)=\sum_{i=1}^n \sbn(F_i)=n$.
 
One can also define \emph{deletion backdoor sets} of a CNF formula $F$
with respect to a base class of formulas by requiring that deleting
all literals $x,\neg x$ with $x\in B$ from~$F$ produces a formula that
belongs to the base class~\cite{NishimuraRagdeSzeider07}. For many
base classes it holds that every deletion backdoor set is a strong
backdoor set, but in most cases, including the base class \Nested, the
reverse is not true. In fact, it is easy to see that if a CNF formula
$F$ has a \Nested-deletion backdoor set of size $k$, then
$\tw^*(F)\leq k+3$. In other words, the parameter ``size of a smallest
deletion \Nested-backdoor set'' is dominated by the parameter
incidence treewidth and therefore of limited interest. We note in
passing, that one can use the algorithm from \cite{MarxSchlotter12} to
show that the detection of deletion \Nested-backdoor sets is
fixed-parameter tractable.

\paragraph{Related Work.}
Williams \etal~\cite{WilliamsGomesSelman03} introduced the notion of \BDSs
to explain favorable running times and the heavy-tailed behavior of
SAT and CSP solvers on practical instances.
The parameterized complexity of finding small backdoor sets was
initiated by Nishimura \etal~\cite{NishimuraRagdeSzeider04-informal}
who showed that with respect to the classes of Horn formulas and of
2CNF formulas, the detection of strong backdoor sets is
fixed-parameter tractable. Their algorithms exploit the fact that for
these two base classes strong and deletion backdoor sets coincide.
For other base classes, 
deleting literals is a less
powerful operation than applying partial truth assignments. This is
the case for the class \Nested but also for the class \RHorn of
renamable Horn formulas. In fact, finding a deletion \RHorn-backdoor
set is fixed-parameter tractable~\cite{RazgonOSullivan09}, but it is
open whether this is the case for the detection of strong
\RHorn-backdoor sets. For clustering formulas the situation is
similar: detection of deletion backdoor sets is fixed-parameter
tractable, detection of strong backdoor sets is most probably not
\cite{NishimuraRagdeSzeider07}. Very recently, the authors of the
present paper showed that for the base class of formulas whose
incidence graph is acyclic there is a fixed-parameter approximation
algorithm for strong backdoor sets. That is, the following problem is
fixed-parameter tractable: find a strong backdoor set of size at most
$k$ or decide that there is no strong backdoor set of size at most
$2^k$~\cite{GaspersSzeider11a}. The present paper extends the ideas
from \cite{GaspersSzeider11a} to the significantly more involved case
with \Nested as the base class.

We conclude this section by referring to a recent survey on the
parameterized complexity of backdoor sets~\cite{GaspersSzeider11festschrift}.

\section{Preliminaries}
\label{section:prelims}

\paragraph{Parameterized Complexity.}
Parameterized Complexity \cite{DowneyFellows99,FlumGrohe06,Niedermeier06} is a two-di\-men\-sio\-nal framework to classify the complexity of problems based on their
input size $n$ and some additional parameter $k$.
It distinguishes between running times of the form $f(k) n^{g(k)}$ where the degree
of the polynomial depends on $k$ and running times of the form $f(k) n^{O(1)}$ where the exponential part of the running time is independent of $n$.

A parameterized problem is
\emph{fixed-parameter tractable} (\FPT) if there exists an algorithm that
solves an input of size $n$ and parameter $k$ in time bounded by $f(k)
n^{O(1)}$.  In this case we say that the \emph{parameter
  dependence} of the algorithm is $f$ and we call it an 
\emph{\FPT\ algorithm}.

Parameterized Complexity has a hardness theory, similar to the theory
of \NP-completeness to show that certain problems have no \FPT\
algorithm under complexity-theoretic assumptions.


\paragraph{Graphs.}
Let $G=(V,E)$ be a simple, finite graph.
Let $S \subseteq V$\longversion{ be a subset of its vertices} and $v\in V$\longversion{ be a vertex}.
We denote by $G - S$ the graph obtained from $G$ by removing all vertices in $S$ and all edges incident to vertices in $S$.
We denote by $G[S]$ the graph $G - (V\setminus S)$.
The \emph{(open) neighborhood} of $v$ is $N(v) = \set{u\in V : uv\in E}$, the \emph{(open) neighborhood} of $S$ is $N(S) = \bigcup_{u\in S}N(u)\setminus S$, and their \emph{closed
neighborhoods} are $N[v] = N(v)\cup \set{v}$ and $N[S] = N(S)\cup S$, respectively.
A $v_1$--$v_k$ \emph{path} $P$ of length $k$ in $G$ is a sequence of $k$ pairwise distinct vertices $(v_1, v_2, \cdots, v_k)$ such that $v_i v_{i+1}\in E$ for each $i\in \{1, \dots, k-1\}$.
The vertices $v_1$ and $v_k$ are the \emph{endpoints} of $P$ and all other vertices from~$P$ are \emph{internal}. An edge is \emph{internal} to
$P$ if it is incident to two internal vertices from $P$.
Two or more paths are independent if none of them contains an inner
vertex of another. 
%
\longversion{

A \emph{tree decomposition} of $G$ is a pair 
$(\{X_i : i\in I\},T)$
where $X_i \subseteq V$, $i\in I$, and $T$ is a tree with elements
of $I$ as nodes
such that:
\begin{enumerate}
  \item $\bigcup_{i\in I} X_i = V$;
  \item $\forall uv\in E$, $\exists i \in I$ such that $\{u,v\} 
\subseteq X_i$;
  \item $\forall i,j,k \in I$, if $j$ is on the path from $i$ to $k$ in $T$ 
then $X_i \cap X_k \subseteq X_j$.
\end{enumerate}
The \emph{width} of a tree decomposition is $\max_{i \in I} |X_i|-1$. }%
The \emph{treewidth} \cite{RobertsonSeymour86} of $G$
\longversion{is the minimum width taken over all tree decompositions
of $G$ and it }is denoted by $\tw(G)$.
A graph is \emph{planar} if it can be drawn in the plane with no crossing edges.
For other standard graph-theoretic notions not defined here, we refer to \cite{Diestel00}.

\paragraph{CNF Formulas and Satisfiability.}
We consider propositional formulas in conjunctive normal form (CNF) where no clause contains
a complementary pair of literals.
For a clause $c$, we write $\lit(c)$ and $\var(c)$ for the sets of literals and variables
occurring in $c$, respectively.
For a CNF formula $F$ we write $\cla(F)$ for its set of clauses,
$\lit(F) = \bigcup_{c\in \cla(F)} \lit(c)$ for its set of literals, and
$\var(F) = \bigcup_{c\in \cla(F)} \var(c)$ for its set of variables.

For a set $X\subseteq \var(F)$ we denote by $2^X$ the set of
all mappings $\tau:X\rightarrow \set{0,1}$, the \emph{truth assignments} on $X$.
A truth assignment on $X$
can be extended to 
the literals over $X$ 
by setting $\tau(\neg x) = 1-\tau(x)$ for all $x\in X$.
Given a CNF formula $F$ and a truth assignment $\tau \in 2^X$ we define
$F[\tau]$ to be the formula obtained from $F$ by removing all clauses $c$
such that $\tau$ sets a literal of $c$ to~1, and removing the literals set to~0
from all remaining clauses.

A CNF formula $F$ is \emph{satisfiable} if there is some $\tau\in
2^{\var(F)}$ with $F[\tau]=\emptyset$. 
SAT is the $\NP$-complete problem of deciding whether a given CNF formula is
satisfiable~\cite{Cook71,Levin73}. \#SAT is the \#P-complete problem of
determining the number of distinct $\tau\in 2^{\var(F)}$ with $F[\tau]=\emptyset$ \cite{Valiant79b}.

\paragraph{Nested Formulas.}
Consider a linear order $<$ of the variables of a CNF formula $F$.
A clause $c$ \emph{straddles} a clause $c'$ if there are variables $x,y \in \var(c)$ and $z\in \var(c')$
such that $x<z<y$. Two clauses \emph{overlap} if they straddle each other.
A CNF formula $F$ is \emph{nested} if there exists a linear ordering $<$ of $\var(F)$ in which no two clauses of $F$ overlap
\cite{Knuth90}.
The satisfiability of a nested CNF formula can be determined in polynomial time \cite{Knuth90}.

The \emph{incidence graph} of a CNF formula $F$ is the bipartite graph $\inc(F)=(V,E)$ with
$V = \var(F) \cup \cla(F)$ and for a variable $x \in \var(F)$ and a clause $c \in \cla(F)$
we have $x c \in E$ if $x\in \var(c)$. The \emph{sign} of the edge $x c$ is \emph{positive}
if $x\in \lit(c)$ and \emph{negative} if $\neg x \in \lit(c)$.

The graph $\incu(F)$ is $\inc(\univ(F))$, where $\univ(F)$ is obtained from $F$ by adding a \emph{universal} clause $c^*$ containing all variables of $F$.
By a result of Kratochv\'{\i}l and K\v{r}iv\'{a}nek \cite{KratochK93}, $F$ is nested \myiff $\incu(F)$ is planar.
Since $\inc(F)$ has treewidth at most $3$ if $F$ is nested \cite{BiedlH04}, the number of satisfying assignments of
$F$ can also be counted in polynomial time \cite{FischerMakowskyRavve06,SamerSzeider10}.

\paragraph{Backdoors.}
Backdoor sets are defined with respect to a fixed class $\cC$ of CNF
formulas, the \emph{base class}.
Let $B$ be a set of propositional variables and $F$ be a CNF formula.
$B$ is a \emph{strong} \emph{$\cC$-\BDS{}} of $F$ if $F[\tau]\in \cC$ for each $\tau \in 2^B$.
$B$ is a \emph{deletion $\cC$-\BDS{}} of $F$ if $F - B \in \cC$, where $F - B = \set{C \setminus \set{x, \neg x : x\in B} : C \in F}$.

If we are given a strong $\cC$-\BDS of $F$ of
size $k$, we can reduce the satisfiability of $F$ to the satisfiability
of $2^k$ formulas in $\cC$.
Thus SAT becomes \FPT\ in $k$.
If $\cC$ is clause-induced (\ie, $F\in \cC$ implies $F'\in \cC$ for every $F'\subseteq F$),
any deletion $\cC$-\BDS of $F$
is a strong $\cC$-\BDS of $F$.
The interest in deletion \BDSs is motivated for base classes where they
are easier to detect than strong \BDSs.
The challenging problem is to find a strong 
or deletion
$\cC$-\BDS of size at most $k$ if it exists.
Denote by $\sbn(F)$ the size of a smallest \SNBDS.


\paragraph{Minors and Grids.}
%
%
%
%
The \emph{$r$-grid} is the graph
$L_r=(V,E)$ with vertex set $V = \{(i, j) : 1 \le i \le r$, $1 \le j \le r\}$ in which two vertices
$(i,j)$ and $(i',j')$ are adjacent \myiff $|i-i'|+|j-j'|=1$.
We say that a vertex $(i,j)\in V$ has horizontal index $i$ and vertical index $j$.

A graph $H$ is a \emph{minor} of a graph $G$ if $H$ can be obtained from a subgraph of $G$
by contracting edges. The \emph{contraction} of an edge $uv$
makes $u$ adjacent to all vertices in $N(v)\setminus \{u\}$ and removes $v$.

If $H$ is a minor of $G$, then one can find a model of $H$ in $G$.
A \emph{model} of $H$ in $G$ is a set of vertex-disjoint connected subgraphs
of $G$, one subgraph $C_u$ for each vertex $u$ of $H$, such that if $uv$ is an edge of $H$, then
there is an edge of $G$ with one endpoint in $C_u$ and the other in $C_v$.

By Wagner's theorem \cite{Wagner37}, a graph is planar \myiff it has no
$K_{3,3}$ and no $K_5$ as a minor. Here, $K_5$ denotes the complete graph
on $5$ vertices and $K_{3,3}$ the complete bipartite graph with $3$ vertices in
both independent sets of the bipartition.

We will use  Robertson and Seymour's grid-minor theorem.

\begin{theorem}[\cite{RobertsonSeymour86b}]
 For every positive integer $r$, there exists a constant $f(r)$ such that if a graph~$G$
 has treewidth at least $f(r)$, then $G$ contains an $r$-grid as a minor.
\end{theorem}

\noindent
By \cite{RobertsonSeymourThomas94}, $f(r) \le 20^{2 r^5}$.
A linear FPT algorithm (parameterized by $k$) by Bodlaender
\cite{Bodlaender96} finds a tree decomposition of width at most $k$ of a graph $G$ if $\tw(G)\le k$.
A quadratic FPT algorithm (parameterized by $r$) by Kawarabayashi \etal~\cite{KawarabayashiKR12} finds an $r$-grid minor in a graph $G$ if $G$ contains an $r$-grid as a minor.%
\shortversion{

Proofs of statements marked with $(\star)$ can be found in the appendix.
}

\section{Detection of Strong Nested-Backdoor Sets}

Our overall approach to find \SNBDSs resembles the approach from \cite{GaspersSzeider11a}
to find strong \textsc{Forest}-\BDSs. Looking more closely at both algorithms,
the reader will see significant differences in how the two main cases are handled.

Let $F$ be a CNF formula and $k$ be an integer.
Our FPT algorithm will decide the satisfiability of $F$ if $F$ has a \SNBDS of size at most $k$.

The first step of the algorithm is to find a good approximation for a smallest \SNBDS.
Specifically, it will either determine that $F$ has no \SNBDS of size at most~$k$, or
it will compute a \SNBDS of size at most $2^k$. In case 
$F$ has no \SNBDS of size at most $k$, the algorithm stops, and if
it finds a \SNBDS $B$ of size at most $2^k$, it uses Knuth's algorithm \cite{Knuth90}
to check for every assignment
$\tau \in 2^B$ whether $F[\tau]$ is satisfiable and answers \textsc{Yes} if at least one
such assignment reduced $F$ to a satisfiable formula and \textsc{No} otherwise.
Since $\tw^*(F[\tau])\le 3$ \cite{BiedlH04} for every truth assignment $\tau$ to $B$,
a tree decomposition of $\inc(F[\tau])$ can be computed in linear time \cite{Bodlaender96}, and
treewidth-based dynamic programming algorithms can be used to compute the number of satisfying assignments
of $F[\tau]$ in polynomial time \cite{FischerMakowskyRavve06,SamerSzeider10}.

We will arrive at our main theorem.

\begin{theorem}\label{thm:sat}
The problems SAT and  \#SAT are \fpt parameterized by $\sbn(F)$.
\end{theorem}

It only remains to find a \SNBDS with an \FPT\ algorithm.  In the
remainder of this section we present an \FPT\ algorithm that either
determines that $F$ has no \SNBDS of size at most $k$, or computes one
of size at most $2^k$. An algorithm of that kind is called an
\emph{\FPT-approximation algorithm}~\cite{Marx08b}, as it is an \FPT\
algorithm that computes a solution that approximates the optimum with an
error bounded by a function of the parameter.



Consider the incidence graph $G=(V,E)=\inc(F)$ of $F$. By \cite{RobertsonSeymourThomas94},
it either has treewidth at most $\ktw$, or it has a $\kgrid$-grid as a minor.
Here,
\begin{align*}
 \ktw &:= 20^{2 \kgrid^5},\\
 \kgrid &:= 4 \cdot \sqrt{\kobs+1},\\
 \kobs &:= 2^{k} \cdot \ksame + k,\text{ and}\\
 \ksame &:= 15 \cdot 2^{2k+2}.
\end{align*}

\subsection{Large Grid Minor}

\tikzset{var/.style={inner sep=.15em,circle,fill=black,draw},
         clause/.style={minimum size=1mm,rectangle,fill=white,draw},
         label distance=-2pt}

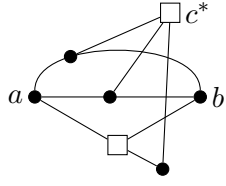
\begin{figure}[tb]
 \centering
  \begin{tikzpicture}[xscale=1,yscale=0.8]
   \node (a) at (0,0) [var,label=left:$a$] {};
   \node (b) at (2.2,0) [var,label=right:$b$] {};
   
   \draw (a)--(b) node (d) [pos=0.45,var] {};
   \draw (a) .. controls +(0,1) and +(0,1) .. (b) node (e) [pos=0.3,var] {};
   \node (f1) at (1.1,-0.8) [clause] {};
   \draw (a)--(f1)--(b);
   \node (f) at (1.7,-1.2) [var] {};
   \draw (f1)--(f);
   
   \node (c) at (1.8,1.4) [clause,label=right:$c^*$] {};
   \draw (d)--(c)--(e) (c)--(f);
  \end{tikzpicture}
  \caption{A \obs leading to a $K_{3,3}$-minor with the universal clause $c^*$.}
  \label{fig:obs}
\end{figure}

The goal of this subsection is to design an \FPT\ algorithm that, given 
a $\kgrid$-grid as a minor in~$G$, computes a set $S^*$ of $2^{O(k^{10})}$ variables
from $\var(F)$ such that every \SNBDS of size at most $k$ contains a variable from $S^*$.
Suppose $G$ has a $\kgrid$-grid as a minor.

\begin{definition}
 An $a$--$b$ \emph{\obs\/} is a subgraph of $\inc(F)$ consisting of
\begin{itemize}
 \item five distinct vertices $a,b,p_1,p_2,p_3$, such that $p_1,p_2,p_3$ are variables,
 \item three independent $a$--$b$ paths $P_1,P_2,P_3$, and
 \item an edge between $p_i$ and a vertex from $P_i$ for each $i\in \{1,2,3\}$.
\end{itemize}
\end{definition}
In particular, if a path $P_i$ has a variable $v$ as an interior vertex, we can take $p_i:=v$.
See Figure~\ref{fig:obs}.

\begin{lemma}\label{lem:obs}
 If $F'$ is a CNF formula such that $\inc(F')$ contains a \obs, then $F' \notin \Nested$.
\end{lemma}
\begin{proof}
 Suppose $\inc(F')$ contains a \obs.
 We will exhibit a $K_{3,3}$-minor in $\incu(F')$.
 A model of a $K_{3,3}$ can be obtained by taking the subgraphs consisting of the singleton vertices $a$, $b$, and the universal clause $c^*$
 for one side of the bipartition, and the three subgraphs induced by $(P_i\cup \{p_i\})\setminus \{a,b\}, 1\le i\le 3,$ for the other side.
 Since, by Wagner's Theorem \cite{Wagner37}, no planar graph has a $K_{3,3}$ as a minor, and by a result of
 Kratochv\'{\i}l and K\v{r}iv\'{a}nek \cite{KratochK93}, $F'$ is nested
 \myiff $\incu(F')$ is planar, we conclude that
 $F' \notin \Nested$.
\myqed \end{proof}

\noindent
By Lemma \ref{lem:obs}, we have that for each assignment to the variables of a 
\SNBDS, at least one variable from each \obs vanishes in the reduced formula.
Using the $r$-grid, we now find a set $\cO$ of $\kobs$ vertex-disjoint \obss in $G$.

\begin{restatable}{lemma}{LemGrid}\label{lem:grid}\shortversion{\textup{($\star$)}}
 Given a $\kgrid$-grid minor of $G=\inc(F)$, a set of $\kobs$ vertex-disjoint \obss can be found in polynomial time.
\end{restatable}
\longversion{\begin{proof}
Let $i$ and $j$ be two integers with $1\le i\le \kgrid/3$ and $1\le j\le \kgrid/4$.
In the $\kgrid$-grid, consider the two vertices $a$ and $b$,
with $a=(3i-1,4j-3)$ and $b=(3i-1,4j)$,
and 3 independent $a$--$b$ paths containing only vertices with
horizontal index between $3i-2$ and $3i$ and vertical index between
$4j-3$ and $4j$. See Figure \ref{fig:grid}.
Denote this subgraph of the grid by $Q$.

%

%

In the $r$-grid model, each vertex $v$ from $Q$ corresponds to a connected subgraph $C_v$.
For each edge $uv$ from $Q$, select one representative edge with one endpoint in $C_u$ and the other endpoint in $C_v$.
For each vertex $v$ from $Q$, compute a spanning tree of $C_v$.
Consider the set $N_v$ of vertices from $C_v$ that are incident to representative edges.
Note that $2\le |N_v|\le 3$.
Select one representative vertex $r_v$ from $C_v$ such that the spanning tree contains independent paths from $r_v$ to each vertex in $N_v$.
If $|N_v|=2$, $r_v$ is any vertex from the unique subpath of the spanning tree connecting the two vertices from $N_v$.
If $|N_v|=3$, $r_v$ is the unique vertex that is on all subpaths of the spanning tree pairwise connecting vertices from $N_v$.
A \obs in $G$ is now formed by connecting each representative vertex to their representative edges by independent paths
chosen as subpaths of the spanning trees.
Note that an $r_a$--$r_b$ \obs in $\cO$ has 3 independent $r_a$--$r_b$ paths with length at least $4$, and thus, each of them contains
at least one variable. Moreover, the \obs contains only vertices from the connected subgraphs of $G$ corresponding to
vertices from the grid with horizontal index between $3i-2$ and $3i$ and vertical index between
$4j-3$ and $4j$.
Thus, any two \obss defined from distinct $(i,j)$ are vertex-disjoint.

The number of \obss defined this way is
$\lfloor \frac{\kgrid}{3} \rfloor \cdot \lfloor \frac{\kgrid}{4} \rfloor \ge (\frac{\kgrid}{3}-1) \cdot (\frac{\kgrid}{4}-1)$.
Since $\kgrid\ge 24$ if $k\ge 1$, we have that $\frac{\kgrid}{3}-1 \ge \frac{\kgrid}{4}+1$.
Thus, the number of \obss is at least $(\frac{\kgrid}{4})^2-1 \ge \kobs$.
\myqed \end{proof}
}

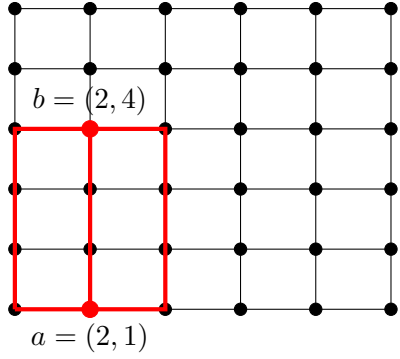
\begin{figure}[tb]
 \centering
  \begin{tikzpicture}[xscale=1,yscale=0.8]
   \pgfmathtruncatemacro\dim{6}
   \foreach \x in {1,2,...,\dim}
    \foreach \y in {1,2,...,\dim} {
     \node at (\x,\y) [var] {};
     \ifnum\x<\dim
       \draw (\x,\y) -- +(1,0);
     \fi
     \ifnum\y<\dim
       \draw (\x,\y) -- +(0,1);
     \fi
    }

    \node at (2,1) [var,inner sep=.2em,fill=red,draw=red,label=below:{$a=(2,1)$}] {};
    \node at (2,4) [var,inner sep=.2em,fill=red,draw=red,label=above:{$b=(2,4)$}] {};
    \draw[ultra thick,red] (1,1)--(1,4)--(3,4)--(3,1)--(1,1) (2,1)--(2,4);
  \end{tikzpicture}
  \caption{The $6$-grid and a highlighted \obs.}
  \label{fig:grid}
\end{figure}

\noindent
Denote by $\cO$ a set of $\kobs$ vertex-disjoint \obss obtained via Lemma \ref{lem:grid}.
A backdoor variable can destroy a \obs either because it participates in the \obs, or because
every setting of the variable satisfies a clause that participates in the \obs.

\begin{definition}
Let $x$ be a variable and $O$ a \obs in $G$.
We say that $x$ \emph{kills} $O$ if neither $\inc(F[x = \true])$ nor $\inc(F[x = \false])$ contains $O$ as a subgraph.
We say that $x$ kills $O$ \emph{internally} if $x\in \var(O)$, and that $x$ kills $O$ \emph{externally} if $x$ kills $O$ but does not kill it internally.
In the latter case, $O$ contains a clause $c$ containing $x$ and a clause $c'$ containing $\neg x$ and we say that
$x$ kills $O$ (externally) \emph{in} $c$ and $c'$.
\end{definition}

\noindent
Our algorithm will make a series of $O(1)$ guesses about the \SNBDS, where each guess is made out of a number of choices that is upper bounded by a function of $k$.
At any stage of the algorithm, a \emph{valid} \SNBDS is one that conforms to the guesses that have been made.
For a fixed series of guesses, the algorithm will compute a set $S\subseteq \var(F)$ such that every valid \SNBDS of size at most $k$ contains a variable from $S$. The union of all
such $S$, taken over all possible series of guesses, forms a set $S^*$ and each \SNBDS of size at most $k$ contains a variable from $S^*$.
Bounding the size of each $S$ by a function of $k$ enables us to bound $|S^*|$ by a function of $k$, and $S^*$ can then be used in a bounded search tree algorithm
(see Subsection~\ref{subsec:algo}).

For any \SNBDS of size at most $k$, at most $k$ \obss from $\cO$ are killed internally since they are vertex-disjoint.
The algorithm guesses $k$ \obss from $\cO$ that may be killed internally.
Let $\cO'$ denote the set of the remaining \obss, which need to be killed externally.

Suppose $F$ has a \SNBDS $B$ of size $k$ killing no \obs from $\cO'$ internally. Then, $B$ defines a partition of $\cO'$ into
$2^k$ parts where for each part, the \obss contained in this part are killed externally by the same set of variables from $B$.
Since $|\cO'| = \kobs - k = 2^k \cdot \ksame$, at least one of these parts contains at least $\ksame$ \obss from $\cO'$.
The algorithm guesses a subset $\cO_s \subseteq \cO'$ of $\ksame$ \obs from this part and it guesses how many variables from the
\SNBDS kill the obstructions in this part externally.

Suppose each \obs in $\cO_s$ is killed externally by the same set of~$\ell$ backdoor variables, and no other backdoor variable
kills any \obs from $\cO_s$. Clearly, $1\le \ell \le k$.
Compute the set of external killers for each \obs in $\cO_s$. Denote by $Z$ the common external killers of the \obs in $\cO_s$.
The presumed \BDS contains exactly $\ell$ variables from $Z$ and no other variable from the \BDS kills any \obs from~$\cO_s$.

We will define three rules for the construction of $S$, and the algorithm will execute the first applicable rule.

\begin{myrule}[Few Common Killers]\label{rule:fewkillers}
 If $|Z|<|\cO_s|$, then set $S:=Z$.
\end{myrule}
The correctness of this rule follows since any valid \SNBDS contains $\ell$ variables from $Z$ and $\ell\ge 1$.

For each $O\in \cO_s$ we define an auxiliary graph $G_O=(Z,E_O)$ whose edge set is initially empty.
As long as $G_O$ has a vertex $v$ with degree $0$ such that $v$ and some other vertex in $Z$
have a common neighbor from $O$ in $G$, select a vertex $u$ of minimum degree in $G_O$ such
that $u$ and $v$ have a common neighbor from $O$ in $G$ and add the edge $uv$ to $E_O$.
As long as $G_O$ has a vertex $v$ with degree $0$, select a vertex $u$ of minimum degree in $G_O$ such that
$v$ has a neighbor $v'\in V(O)$ in $G$ and $u$ has a neighbor $u'\in V(O)$ in $G$
and there is a $v'$--$u'$ path in $O$ in which no internal vertex is adjacent to a vertex from $Z\setminus \{v\}$;
add the edge $uv$ to $E_O$.

By the construction of $G_O$, we have the following property on the degree of all vertices.

\begin{fact}
 For each $O\in \cO_s$, the graph $G_O$ has minimum degree at least $1$. 
\end{fact}

Recall that no clause contains complimentary literals.
Consider two variables $u,v\in Z$ that share an edge in $G_O$.
By the construction of $G_O$, there is a $u$--$v$ path $P$ in $G$ whose internal edges are in $O$, such
that for each variable $z\in Z$, all edges incident to $z$ and a clause from $P$ have the same sign.
Moreover, since no variable from a valid \SNBDS kills $O$ externally, unless it is in $Z$,
for each potential backdoor variable $x\in \var(F) \setminus Z$, all edges incident to $x$ and a clause from $P$ have the same sign.
Thus, we have the following fact.

\begin{fact}\label{fact:2}
 If $u,v\in Z$ share an edge in $G_O$, then for every valid \SNBDS that does not contain $u$ and $v$, there is a truth assignment
 $\tau$ to $B$ such that $\inc(F[\tau])$ contains a $u$--$v$ path whose internal edges are in $O$.
\end{fact}

Consider the multigraph $G_m(\cO_s)= (Z, \biguplus_{O\in \cO_s} E_O)$, \ie, the union of all $G_O$ over all
$O\in \cO_s$, where the multiplicity of an edge is the number of distinct sets $E_O$ where it appears, $O\in \cO_s$.

\begin{myrule}[Multiple Edges]\label{rule:multi}
 If there are two vertices $u,v\in Z$ such that $G_m(\cO_s)$ has a $u$--$v$ edge with multiplicity at
 least $2 \cdot 2^{k}+1$, then set $S:=\{u,v\}$.
\end{myrule}
Consider any valid \SNBDS $B$ of size $k$. Then, by Fact \ref{fact:2}, for each $u$--$v$ edge
there is some truth assignment $\tau$ to $B$ such that $\inc(F[\tau])$ contains a $u$--$v$ path in $G$. Moreover, since each
$u$--$v$ edge comes from a different $O\in \cO_s$, all these $u$--$v$ paths are independent.
Since there are $2^k$ truth assignments to $B$ but at least $2 \cdot 2^{k}+1$ $u$--$v$ edges, for at least one
truth assignment $\tau$ to $B$, there are 3 independent $u$--$v$ paths $P_1,P_2,P_3$ in $\inc(F[\tau])$.
We obtain a $u$--$v$ \obs choosing as $p_i, 1\le i\le 3$, a variable from $P_i$ or a variable neighboring a clause
from $P_i$ and belonging to the same \obs in $\cO_s$.
Thus, any valid \SNBDS contains $u$ or $v$.

Now, consider the graph $G(\cO_s)$ obtained from the multigraph $G_m(\cO_s)$ by merging multiple edges,
\ie, we retain each edge only once.

\begin{myrule}[No Multiple Edges]\label{rule:nomulti}
 Set $S$ to be the $2k$ vertices of highest degree in $G(\cO_s)$ (ties are broken arbitrarily).
\end{myrule}
For the sake of contradiction, suppose $F$ has a valid \SNBDS $B$ of size~$k$ with $B\cap S = \emptyset$.
First, we show a lower bound on the number of edges in $G(\cO_s) - B$.
Since $G_m(\cO_s)$ has at least $\frac{|Z|}{2} \ksame$ edges and each edge has multiplicity at most $2^{k+1}$, the
graph $G(\cO_s)$ has at least $\frac{|Z| \ksame}{2 \cdot 2^{k+1}} = 3\cdot 5\cdot 2^k \cdot |Z|$ edges.
Let $d$ be the sum of the degrees in $G(\cO_s)$ of the vertices in $B \cap Z$.
Now, the sum of degrees of vertices in $S$ is at least $2d$ in $G(\cO_s)$,
and at least $d$ in $G(\cO_s) - B$. Therefore, $G(\cO_s) - B$ has at least $d/2$ edges.
On the other hand, the number of edges deleted to obtain $G(\cO_s) - B$ from $G(\cO_s)$ is at most $d$.
It follows that the number of edges in $G(\cO_s) - B$ is at least a third the number of edges in $G(\cO_s)$,
and thus at least $5\cdot 2^k \cdot |Z|$.

Now, we iteratively build a truth assignment $\tau$ for $B$. Set $H:=G(\cO_s) - B$.
Order the variables of $B$ as $b_1, \dots, b_k$. For increasing~$i$, we set
$\tau(b_i)=0$ if in $G$, the vertex $v\in B$ is adjacent with a positive edge to more paths that correspond to an edge in $H$
than with a negative edge and set $\tau(b_i)=1$ otherwise; if $\tau(b_i)=0$, then remove each edge from $H$ that corresponds to a path
in $G$ that is adjacent with a negative edge to $b_i$, otherwise remove each edge from $H$ that corresponds to a path
in $G$ that is adjacent with a positive edge to $b_i$.

Observe that for a variable $v\in B$ and a path $P$ in $G$ that corresponds to an edge in $G(\cO_s) - B$,
$v$ is not adjacent with a positive and a negative edge to $P$. If $v\in Z$ this follows by the construction of $G_O$,
and if $v\notin Z$, this follows since $v$ does not kill any \obs from $\cO_s$.
Therefore, each of the $k$ iterations building the truth assignment $\tau$ has removed at most half the edges of $H$.
In the end, $H$ has at least $5 |Z|$ edges.

Next, we use the following theorem of Kirousis \etal~\cite{KirousisSS93}.

\begin{theorem}[\cite{KirousisSS93}]
 If a graph has $n$ vertices and $m>0$ edges, then it has an induced subgraph that is $\lceil \frac{m+n}{2n} \rceil$-vertex-connected.
\end{theorem}

\noindent
We conclude that $H$ has an induced subgraph $H'$ that is $3$-vertex-connected.
Let $x,y\in V(H')$. We use Menger's theorem \cite{Menger27}.

\begin{theorem}[\cite{Menger27}]
 Let $G=(V,E)$ be a graph and $x,y\in V$. Then the size of a minimum $x,y$-vertex-cut in $G$ is
 equal to the maximum number of independent $x$--$y$ paths in $G$.
\end{theorem}

\noindent
Since the minimum size of an $x,y$-vertex cut is at least $3$ in $H'$, there are 3 independent
$x$--$y$ paths in $H'$.
Replacing each edge by its corresponding path in $G$, gives rise to 3 walks from $x$ to $y$ in $G$.
Shortcutting cycles, we obtain three $x$--$y$ paths $P_1,P_2,P_3$ in~$G$.
By construction, each edge of these paths is incident to a vertex from a \obs in~$\cO_s$.
We assume that $P_1,P_2,P_3$ are edge-disjoint.
Indeed, by the construction of the $G_O$, $O\in \cO_s$, they can only share the first and last edges.
In case $P_1$ shares the first edge with $P_2$, replace $x$ by its neighbor on $P_1$,
remove the first edge from $P_1$ and $P_2$, and replace $P_3$ by its symmetric difference with this edge.
Act symmetrically for the other combinations of paths sharing the first or last edge.

\begin{restatable}[\cite{Gaspers12}]{lemma}{LemPaths}\label{lem:paths}\shortversion{\textup{($\star$)}}
 Let $G=(V,E)$ be a graph. If there are two vertices $x,y\in V$ with 3 edge-disjoint $x$--$y$ paths in $G$,
  then there are two vertices $x',y'\in V$ with 3 independent $x'$--$y'$ paths in $G$.
\end{restatable}
\longversion{\begin{proof}
Let $P_1,P_2,P_3$ denote 3 edge-disjoint $x$--$y$ paths,
and let $S=\{s_1,s_2,s_3\}$, where $s_i$ neighbors $x$ on $P_i$.
Consider the connected component $G'$ of $G - \{x\}$ containing $y$.
Then, $G'$ contains all vertices from $S$.
Let $T$ be a spanning tree of $G'$.
Select $y'$ to be the vertex belonging to every subpath of $T$ that has
two vertices from $S$ as endpoints.
Set $x':=x$, and obtain 3 independent $x'$--$y'$ paths in $G$ by moving
from $x'$ to $s_i$, and then along the $s_i$--$y'$ subpath of $T$ to $y'$,
$1\le i\le 3$.
\myqed \end{proof}

}

\noindent
By Lemma \ref{lem:paths} we obtain two vertices $x',y'$ in $G$ with 3 independent $x'$--$y'$ paths $P_1',P_2',P_3'$ in $G$.
Since the lemma does not presuppose any other edges in $G$ besides those from the edge-disjoint $x$--$y$ paths,
$P_1',P_2',P_3'$ use only edges from the paths $P_1,P_2,P_3$. Thus, each edge of $P_1',P_2',P_3'$ is incident to
a vertex from a \obs in $\cO_s$.
Thus, we obtain a $x'$--$y'$ \obs with the paths $P_1',P_2',P_3'$, and for each path $P_i'$, we choose a variable from this path
or a variable from~$\cO_s$ neighboring a clause from this path.
We arrive at a contradiction for $B$ being a valid \SNBDS.
This proves the correctness of Rule \ref{rule:nomulti}.

\medskip

The number of possible guesses the algorithm makes is upper bounded by $\binom{\kobs}{k} \cdot \binom{\kobs-k}{\ksame}\cdot k = 2^{O(k^8)}$,
and each series of guesses leads to a set $S$ of at most $\ksame$ variables. Thus, the set $S^*$, the union of all such $S$,
contains at most $2^{O(k^8)}\cdot \ksame = 2^{O(k^{10})}$ variables.
Finally, we have shown the following lemma in this subsection.

\begin{lemma}\label{lem:wall}
  There is an \FPT\ algorithm that, given a CNF formula~$F$, a positive
  integer parameter~$k$, and a $\kgrid$-grid as a minor in $\inc(F)$,
  computes a set $S^*\subseteq \var(F)$ of size $2^{O(k^{10})}$ such
  that every \SNBDS of size at most $k$ contains a variable from $S^*$.
\end{lemma}

%

\subsection{Small Treewidth}

The goal of this subsection is to design an \FPT\ algorithm that, given 
a tree decomposition of $G$ of width at most $\ktw$, finds a \SNBDS of $F$ of size $k$ or determines that $F$ has no such \SNBDS.

Our algorithm uses Arnborg \etal's extension \cite{ArnborgLagergrenSeese91} of Courcelle's Theorem \cite{Courcelle90}.
It gives, amongst others, an \FPT\ algorithm that takes as input a graph $\mathcal{A}$ with labeled vertices and edges and
a Monadic Second Order (MSO) sentence $\varphi(X)$, and computes a minimum-sized set of vertices $X$ such that $\varphi(X)$ is true in $\mathcal{A}$.
Here, the parameter is $|\varphi|+\tw(\mathcal{A})$.

%

We will define a labeled graph whose treewidth is upper bounded by a function of $\tw(\inc(F))$, and an
MSO-sentence of constant length such that the graph models the MSO-sentence
whenever its argument is a \SNBDS of $F$.
 
\begin{restatable}{lemma}{LemTw}\label{lem:tw}\shortversion{\textup{($\star$)}}
 There is an \FPT\ algorithm that takes as input a CNF formula~$F$, a positive integer parameter $k$, and a tree decomposition of $G$ of width at most $\ktw$,
 and finds a \SNBDS of $F$ of size $k$ if one exists.
\end{restatable}
\longversion{\begin{proof}
First, we define the labeled graph $A_F$ for $F$.
The set of vertices of $A_F$ is $\text{LIT} \cup \text{CLA}$, with $\text{LIT} = \lit(\univ(F))$ and $\text{CLA} = \cla(\univ(F))$.
They are labeled by $\text{LIT}$ and $\text{CLA}$, respectively.
The vertices from $\var(F)$ are additionally labeled by $\text{VAR}$.
The subset of edges $\set{x \neg x: x\in \var(F)}$ is labeled $\text{NEG}$ and the subset of edges
$\{x c: x\in \text{LIT}, \allowbreak c\in \text{CLA},$ $x\in \lit(c)\}$ is labeled $\text{IN}$.

Since a tree decomposition for $A_F$ may be obtained from a tree decomposition for $\inc(F)$
by replacing each variable by both its literals and adding the universal clause $c^*$ to each bag of
the tree decomposition, we have that $\tw(A_F)\le 2 \cdot \tw(\inc(F))+2$.

The goal is to find a subset $X$ of variables such that for each truth assignment $\tau$ to $X$ the incidence graph of $\univ(F[\tau])$ is planar.
By Wagner's theorem \cite{Wagner37}, a graph is planar \myiff it has no $K_5$ and no $K_{3,3}$ as a minor.

We break up our MSO sentence into several simpler sentences and we use the notation of \cite{FlumGrohe06}.

The following sentence checks whether $X$ is a subset of variables.
\begin{align*}
 \text{var}(X) = \forall x (Xx \rightarrow \text{VAR} x)
\end{align*}
An assignment to $X$ is a subset $Y$ of $\text{LIT}$ containing no complementary literals such that every selected literal
is a variable from $X$ or its negation, and for every variable $x$ from $X$, $x$ or $\neg x$ is in~$Y$.
The following sentence checks whether $Y$ is an assignment to $X$.
\begin{align*}
 \text{ass}(X,Y) &= \forall y (Yy \rightarrow ((Xy \vee (\exists z (Xz \wedge \text{NEG} yz)))\\
   & \quad \quad \quad \quad \quad \quad \: \wedge (\forall z (Yz \rightarrow \neg \text{NEG} yz))))\\
   & \quad \; \wedge \forall x (Xx \rightarrow (Yx \vee \exists y (Yy \wedge \text{NEG} xy)))
\end{align*}
To test whether a graph has a $K_5$-minor, we will check whether it contains five disjoint sets of vertices, such that each such set induces a connected subgraph
and all 5 sets are pairwise connected by an edge. Deleting all vertices that are in none of the five sets, and contracting each of the five sets into one vertex,
one obtains a $K_5$. The following sentence checks whether $A$ is disjoint from~$B$.
\begin{align*}
 \text{disjoint}(A,B) = \neg \exists x (Ax \wedge Bx)
\end{align*}
To check whether $A$ is connected with respect to the edges labeled $\text{IN}$, we check that there is no set $B$ that is a proper nonempty subset of $A$ such that $B$
is closed under taking neighbors in $A$.
\begin{align*}
 \text{connected}(A) = \neg \exists B (&\exists x (Ax \wedge \neg Bx) \wedge \exists x (Bx) \wedge \forall x (Bx \rightarrow Ax)\\
                       & \forall x,y ((Bx \wedge Ay \wedge \text{IN}xy) \rightarrow By))
\end{align*}
The following sentence checks whether some vertex from $A$ and some vertex from $B$ have a common edge labeled $\text{IN}$.
\begin{align*}
 \text{edge}(A,B) = \exists x,y (Ax \wedge Bx \wedge \text{IN}xy)
\end{align*}
An assignment removes from the incidence graph all variables that are assigned and all clauses that are assigned correctly.
Therefore, the minors we seek must not contain any variable that is assigned nor any clause that is assigned correctly.
The following sentence checks whether all vertices from a set $A$ survive when assigning $Y$ to $X$.
\begin{align*}
 \text{survives}(A,X,Y) = \neg \exists x (Ax \wedge (Xx \vee \exists y (Yy \wedge \text{IN}yx)))
\end{align*}
Testing whether a $K_5$-minor survives in the incidence graph is now done as follows.
\begin{align*}
 \text{$K_5$-minor}(X,Y) &= \exists A_1, \dots, A_5 ( \bigwedge_{i=1}^5 \text{survives}(A_i) \wedge
                                                     \bigwedge_{1\le i\neq j \le 5} \text{disjoint}(A_i,A_j) \wedge\\
                                                     \quad & \bigwedge_{i=1}^5 \text{connected}(A_i) \wedge
                                                     \bigwedge_{1\le i\neq j \le 5} \text{edge}(A_i,A_j))
\end{align*}
Similarly, testing whether a $K_{3,3}$-minor survives in the incidence graph is done as follows.
\begin{align*}
 K_{3,3}&\text{-minor}(X,Y) = \exists A_1, A_2, A_3, B_1, B_2, B_3 ( \bigwedge_{i=1}^3 (\text{survives}(A_i) \wedge \text{survives}(B_i)) \wedge\\
                                                     & \bigwedge_{1\le i,j \le 3} \text{disjoint}(A_i,B_j) \wedge
                                                      \bigwedge_{1\le i\neq j \le 3} (\text{disjoint}(A_i,A_j) \wedge \text{disjoint}(B_i,B_j)) \wedge\\
                                                     & \bigwedge_{i=1}^3 (\text{connected}(A_i) \wedge \text{connected}(B_i)) \wedge
                                                     \bigwedge_{1\le i,j \le 3} \text{edge}(A_i,B_j))
\end{align*}
Our final sentence checks whether $X$ is a \SNBDS of $F$.
\begin{align*}
 \text{SNB}(X) = \text{var}(X) \wedge \forall Y (\text{ass}(X,Y) \rightarrow \neg (\text{$K_5$-minor}(X,Y) \vee \text{$K_{3,3}$-minor}(X,Y))))
\end{align*}
Since $|\text{SNB}| = O(1)$ and $\tw(A_F) \le 2 \cdot \ktw+2$, by \cite{ArnborgLagergrenSeese91} there is an \FPT\ algorithm finding
a \SNBDS of minimum size, where the parameter is $k$.
\myqed \end{proof}

}

We note that in the case where $\tw(G)$ is bounded, one could immediately solve the satisfiability problem for $F$ \cite{FischerMakowskyRavve06,SamerSzeider10}.
However, finding a \SNBDS enables us to give an \FPT\ approximation algorithm for the backdoor detection problem.

\subsection{The FPT algorithm}
\label{subsec:algo}

Our \FPT-approximation algorithm combines the results from the previous two subsections.
In case $G$ has treewidth at most $\ktw$, Lemma \ref{lem:tw} is used to find a solution of size $k$ if one exists.
Otherwise, Lemma \ref{lem:wall} provides a set $S^*$ of $2^{O(k^{10})}$ variables such that any solution of size at most~$k$ contains
a variable from $S^*$. For each $x\in S^*$, the algorithm recurses on both formulas $F[x=0]$ and $F[x=1]$.
If both recursive calls return \SNBDSs $B_{\neg x}$ and $B_{x}$, then $\{x\} \cup B_{x} \cup B_{\neg x}$
is a \SNBDS of $F$, otherwise, no \SNBDS of $F$ of size at most~$k$ contains $x$.

Since $B_{\neg x}$ could be
disjoint from $B_{x}$ in the worst case, while $F[x=0]$ and $F[x=1]$ could have \SNBDSs $B_{\neg x}'$ and
$B_{x}'$ of size $k-1$ with $B_{\neg x}'=B_{x}'$, our approach approximates the optimum with a factor of $2^k/k$.

\begin{restatable}{theorem}{ThmStrong}\label{thm:strong}\shortversion{\textup{($\star$)}}
There is an \FPT\ algorithm, which, for a CNF formula $F$ and a positive integer parameter $k$, either concludes that $F$
has no \SNBDS of size at most $k$ or finds a \SNBDS of $F$ of size at most $2^k$.
\end{restatable}
\longversion{\begin{proof}
 If $k\le 1$, our algorithm solves the problem exactly in polynomial time.
 Otherwise, it runs 
 Bodlaender's \FPT\ algorithm \cite{Bodlaender96} with input $G$ and parameter $\ktw$ to either find a tree decomposition of $G$ of width at most
 $\ktw$ or to determine that $\tw(G)>\ktw$. 
 In case a tree decomposition of width at most $\ktw$ is found, the algorithm uses Lemma \ref{lem:tw} to compute a \SNBDS of $F$ of size $k$ if one exists,
 and it returns the answer.

 In case Bodlaender's algorithm determines that $\tw(G)>\ktw$, by \cite{RobertsonSeymourThomas94} we know that $G$ has a $\kgrid$-grid
 as a minor. Such a $\kgrid$-grid is found by running the FPT algorithm of Kawarabayashi \etal \cite{KawarabayashiKR12} with input $G$ and
 parameter $\kgrid$. The algorithm now executes the procedure from Lemma \ref{lem:wall} to find a set $S^*$ of $2^{O(k^{10})}$ variables
 from $\var(F)$ such that every \SNBDS of size at most $k$ contains a variable from $S^*$.
 The algorithm considers all possibilities that the \BDS contains some $x\in S^*$; there are $2^{O(k^{10})}$ choices for $x$.
 For each such choice, recurse on $F[x = \true]$ and $F[x = \false]$ with parameter $k-1$.
 If, for some $x\in S^*$, both recursive calls return \BDSs $B_x$ and $B_{\neg x}$, then return $B_x\cup B_{\neg x}\cup \set{x}$,
 otherwise, return \textsc{No}. As $2^k-1 = 2\cdot (2^{k-1}-1)+1$, the solution size is upper bounded by $2^k-1$.
 On the other hand, if at least one recursive call returns \textsc{No} for every $x\in S^*$, then $F$ has no \SNBDS of size at most $k$.
\myqed \end{proof}

}

\noindent
In particular, this proves Theorem \ref{thm:sat}.

\section{Conclusion}

We have classified the problems SAT and \#SAT as fixed-parameter
tractable when parameterized by the size of a smallest strong backdoor
set with respect to the base class of nested formulas. As argued in the
introduction, this parameter is incomparable with incidence treewidth.
  
The parameter dependence makes our algorithm impractical.  However, we
would like to note that the class of fixed-parameter tractable problems
has proven to be quite robust: Once a problem is shown to belong to this
class, one can start to develop faster and more practical algorithms.
For many cases in the past this was successful. For instance, the
problem of recognizing graphs of genus~$k$ was originally shown to be
fixed-parameter tractable by means of non-constructive tools from graph
minor theory~\cite{FellowsLangston88}. Later a linear-time algorithm
with doubly exponential parameter dependence was found~\cite{Mohar96},
and more recently, the algorithm could be improved to a single
exponential parameter dependence~\cite{KawarabayashiMoharReed08}.

It would be interesting to see whether a similar improvement is possible for
finding or FPT-approximating strong backdoor sets with respect to nested
formulas.

{
\bibliographystyle{plain}
\bibliography{literature}
}

\shortversion{
\newpage
\appendix

\section{Appendix}

\LemGrid*

\LemPaths*

\LemTw*

\ThmStrong*
\begin{proof}
 If $k\le 1$, our algorithm solves the problem exactly in polynomial time.
 Otherwise, it runs 
 Bodlaender's \FPT\ algorithm \cite{Bodlaender96} with input $G$ and parameter $\ktw$ to either find a tree decomposition of $G$ of width at most
 $\ktw$ or to determine that $\tw(G)>\ktw$. 
 In case a tree decomposition of width at most $\ktw$ is found, the algorithm uses Lemma \ref{lem:tw} to compute a \SNBDS of $F$ of size $k$ if one exists,
 and it returns the answer.

 In case Bodlaender's algorithm determines that $\tw(G)>\ktw$, by \cite{RobertsonSeymourThomas94} we know that $G$ has a $\kgrid$-grid
 as a minor. Such a $\kgrid$-grid is found by running the FPT algorithm of Kawarabayashi \etal \cite{KawarabayashiKR12} with input $G$ and
 parameter $\kgrid$. The algorithm now executes the procedure from Lemma \ref{lem:wall} to find a set $S^*$ of $2^{O(k^{10})}$ variables
 from $\var(F)$ such that every \SNBDS of size at most $k$ contains a variable from $S^*$.
 The algorithm considers all possibilities that the \BDS contains some $x\in S^*$; there are $2^{O(k^{10})}$ choices for $x$.
 For each such choice, recurse on $F[x = \true]$ and $F[x = \false]$ with parameter $k-1$.
 If, for some $x\in S^*$, both recursive calls return \BDSs $B_x$ and $B_{\neg x}$, then return $B_x\cup B_{\neg x}\cup \set{x}$,
 otherwise, return \textsc{No}. As $2^k-1 = 2\cdot (2^{k-1}-1)+1$, the solution size is upper bounded by $2^k-1$.
 On the other hand, if at least one recursive call returns \textsc{No} for every $x\in S^*$, then $F$ has no \SNBDS of size at most $k$.
\myqed \end{proof}

}

\end{document}